\documentclass[journal,12pt,onecolumn]{IEEEtran}
\usepackage{pict2e}
\usepackage{xcolor}
\usepackage{amsmath}
\usepackage{nccmath}
\usepackage{amssymb} 
\usepackage{amsthm}
\usepackage[linesnumbered,boxruled]{algorithm2e}
\usepackage{graphicx}
\usepackage[font=footnotesize]{caption}
\usepackage[font=footnotesize]{subcaption}
\usepackage[hidelinks=true]{hyperref}
\graphicspath{{Figs/}}
\theoremstyle{plain} \newtheorem{thm}{Theorem}
\theoremstyle{plain} \newtheorem{lem}{Lemma}
\theoremstyle{definition} \newtheorem{defi}{Definition}
\theoremstyle{definition} \newtheorem{corol}{Corollary}

\title{On Storage Allocation in Cache-Enabled Interference Channels with Mixed CSIT}
\author{\IEEEauthorblockN{
	Mohammad Ali Tahmasbi Nejad\IEEEauthorrefmark{1},
	Seyed Pooya Shariatpanahi\IEEEauthorrefmark{2}, and
	Babak Hossein Khalaj\IEEEauthorrefmark{1}
	}\\
	\IEEEauthorblockA{
		\IEEEauthorrefmark{1}Department of Electrical Engineering, Sharif University of Technology, Tehran, Iran\\
		\IEEEauthorrefmark{2}School of Computer Science, Institute for Research in Fundamental Sciences (IPM), Tehran, Iran\\
		E-mails: m\_tahmasbi@ee.sharif.edu,
		pooya@ipm.ir,
		khalaj@sharif.edu
	}}

\let\oldnl\nl
\newcommand{\nonl}{\renewcommand{\nl}{\let\nl\oldnl}}

\begin{document}
\maketitle

\begin{abstract}
	
Recently, it has been shown that in a cache-enabled interference channel, the storage at the transmit and receive sides are of equal value in terms of Degrees of Freedom (DoF). This is derived by assuming full Channel State Information at the Transmitter (CSIT). In this paper, we consider a more practical scenario, where a training/feedback phase should exist for obtaining CSIT, during which instantaneous channel state is not known to the transmitters. This results in a combination of delayed and current CSIT availability, called \emph{mixed CSIT}. In this setup, we derive DoF of a cache-enabled interference channel with mixed CSIT, which depends on the memory available at transmit and receive sides as well as the training/feedback phase duration. In contrast to the case of having full CSIT, we prove that, in our setup, the storage at the receive side is more valuable than the one at the transmit side. This is due to the fact that \emph{cooperation opportunities} granted by transmitters' caches are strongly based on instantaneous CSIT availability. However, \emph{multi-casting opportunities} provided by receivers' caches are robust to such imperfection.
\end{abstract}


\section{Introduction}

Nowadays, the available bandwidth to transfer data has become one of the most important and costly communication resources due to increasing use of communication networks and significant exponential growth of the Internet traffic load. 
Considering the fact that Internet video traffic has become the main part of all consumer Internet traffic \cite{Cisco}, using novel approaches such as caching seems to be of high significance in reducing bandwidth consumption. Content Delivery Networks (CDNs) are  typical examples of enterprises using the caching techniques to deliver requested content to the consumers \cite{Pathan07}. Simply put, in such methods the geographically distributed storage units cache content during network low-peak hours based on the statistics of demands, given their cache sizes. Then, upon consumer requests arrival, this cached data will be used during network high-peak hours to relieve congestion. Both \emph{cache content placement} and \emph{content delivery} phases need careful design considerations to manage network congestion.

In this paper, we propose a novel caching approach that, in addition to the local caching gain, is able to achieve a global caching gain. This gain derives from jointly optimizing both the caching and delivery phases, ensuring that in the delivery phase several different demands can be satisfied with a single multicast transmission. Since the cache placement is performed without knowledge of the actual demands, in order to achieve this gain the placement must be carefully designed such that these multicasting opportunities are created simultaneously for all possible requests in the delivery phase. We show that this global caching gain is relevant if the aggregate global cache size is large enough compared to the total amount of content. Thus, even though the caches cannot cooperate, the sum of the cache sizes becomes an important system
parameter.

Employing the caching idea in communication networks has been the subject of many recent works. Shanmugam et al. propose using cache-enabled \emph{helpers} in cellular networks to offload traffic from congested base stations \cite{Shanmugam13}. In \cite{Golrezaei14}, the idea of Device-to-Device (D2D) communication with distributed caching has been investigated. Also, in \cite{Bastug14} and \cite{Perabathini15}, the authors analyze the role of proactive edge caching in 5G wireless networks. Finally, the throughput-memory trade-off for cache-enabled ad-hoc networks has been investigated in \cite{Gitzenis13}.

The concept of \emph{Coded Caching} introduced in \cite{Maddah-Ali14} considers the caching problem from an information theoretic perspective, resulting in substantial performance gains. Subsequent works extended the results to the D2D setting \cite{Ji16}, networks with hierarchical caches architecture \cite{Karamchandani16}, multi-server setup \cite{Shariatpanahi16}, cache-enabled interference channels \cite{Naderalizadeh16}, and MISO broadcast channels with delayed CSIT \cite{Zhang15}, \cite{Zhang16}.

In the present study, we consider an interference channel where both transmitters and receivers are equipped with cache memories of limited size. In the cache content placement phase, the caches of both sides are filled with content from a library, without knowing the actual requests of the receivers. In the content delivery phase, by knowing the receivers' requests, the transmitters send proper signals on the wireless medium to fulfill these demands. In contrast to \cite{Naderalizadeh16}, which considers the same setup with full CSIT, here we assume a mixed CSIT setting. In other words, in order to provide CSI at the transmit side, the transmitters send  training symbols at the beginning of each time coherence block. Based on this training process, the CSI is fed back to the transmitters, leading to delay in obtaining CSIT. This model is closer to practical scenarios where CSIT should be provided by some sort of feedback \cite{Love08}.

As shown in \cite{Naderalizadeh16}, with full CSIT assumption, increasing cache size at the transmitters or receivers results in higher DoF. That is due to the fact that, by providing memory at the transmitters, they can cooperate in sending data, and the more the storage size is, the higher the cooperation chances will be. On the other hand, receive-side memories provide multi-casting opportunities as discussed in \cite{Maddah-Ali14}.  However, the interesting result of \cite{Naderalizadeh16} emphasizes that the transmitter and receiver memories are of equal importance in terms of DoF, when there is full CSIT. On the contrary, in this paper with the mixed CSIT model, we prove that the receive-side caches are much more valuable than the ones at transmit side. In fact, the cooperation of transmitters heavily relies on perfect instantaneous CSIT, which diminishes in practical mixed CSIT scenarios. The multi-casting opportunities provided by receive-side caches, however, are not sensitive to such information. This result has an important practical design guideline on storage allocation in interference channels: \emph{push the storage towards end-users}.

Other closely related works to our study are \cite{Shariatpanahi14}, \cite{Maddah-Ali15}, \cite{Zhang15}, and \cite{Wigger16}. The papers \cite{Shariatpanahi14} and \cite{Maddah-Ali15} also consider cache-enabled interference channels, with the difference that they assume memory only at the transmit side. In \cite{Zhang15}, Zhang et al. consider the effect of delayed CSIT in a MISO boradcast channel, which is equivalent to having unlimited memory at the transmit side, similar to \cite{Shariatpanahi16}. Here, in contrast, we consider limited transmit-side memory. Finally, \cite{Wigger16} considers a Wyner model, and that the transmitters are connected via a Cloud RAN, which makes their model different from ours.

The rest of the paper is organized as follows. In Section \ref{sec:System_Model}, we present the system model. In Section \ref{sec:PerfectCSIT}, DoF of a cache-enabled interference channel with full CSIT is derived. In Section \ref{sec:DelayedCSIT}, we introduce DoF of a cache-enabled interference channel with mixed CSIT. Finally, Section \ref{sec:Num} presents numerical illustrations, and concludes the paper.

\noindent {\bf Notation:} $(\cdot)^T$ represents transpose operation. $|\cdot|$ denotes the cardinality of a set, and $[K]\triangleq\{1,2,...,K\}$, for any $K\in\mathbb{N}^+$.

\section{System Model}
\label{sec:System_Model}
\sloppy
We consider a block time-varying memoryless wireless interference channel through which $K_t$ transmitters are connected to $K_r$ receivers. Such channel, at each time instant $\tau$, can be modeled as a $K_r$-by-$K_t$ matrix which is denoted by ${\bf H}(\tau)=\big[h_{i,j}(\tau)\big] \in \mathbb{C}^{K_r\times K_t}$, whose entry in the $i$th row and the $j$th column represents the channel gain from transmitter $j$ to the $i$th receiver, and is an independent complex Gaussian random variable ${\cal CN}(0,1)$. This matrix relates the channel's inputs to the outputs by following equation:
\begin{equation}
{\bf y}(\tau) = {\bf H}(\tau) {\bf x}(\tau) + {\bf z}(\tau),
\end{equation}
in which, ${\bf x}(\tau) = \big(x_1(\tau),...,x_{K_t}(\tau)\big)^T$ is the vector consisting of all the transmitted messages, ${\bf y}(\tau) = \big(y_1(\tau),...,y_{K_r}(\tau)\big)^T$ represents the output vector which includes the received signals, and ${\bf z}(\tau) = \big(z_1(\tau),...,z_{K_r}(\tau)\big)^T \in \mathbb{C}^{K_r\times 1}$ denotes the additive white Gaussian noise (AWGN) with the covariance matrix $\sigma^2{\bf I}_{K_r}$.

Then, we define the channel gain vector to the $k$th receiver as follows:
\begin{equation}
{\bf h}_k(\tau) \triangleq 
\big(h_{k,1}(\tau),...,h_{k,K_t}(\tau)\big)^T
, \quad k=1,2,...,K_r,
\end{equation}
which is the transpose of the $k$th row of matrix ${\bf H}(\tau)$, so the received signal of receiver $k$ based on the channel's input vector is as follows:
\begin{equation}
y_k(\tau) = {\bf h}_k(\tau)^T {\bf x}(\tau) + z_k(\tau).
\end{equation}

We divide the time domain into time slots of equal length, $T_c$, such that the channel state during each time slot is fixed, while it is completely independent for different time slots. In other words, the channel transfer matrix elements do not alter in the middle of any time slot, so a discrete-time process, i.e. ${\bf H}[m]={\bf H}(\tau|(m-1)T_c\le \tau < mT_c)$,  represents channel transfer matrix for all time slots $m=1,2,\dots$. 

In this paper, we assume that the CSI is measured at the receivers from training symbols sent by the transmitters. Then, via a feedback channel, this information is fed back to the transmitters. We assume this training/feedback procedure takes the time of length $T_f$. Therefore, in the first $T_f$ seconds of each time coherence block of $T_c$ seconds, the transmitters do not have any instantaneous CSI, while in the rest of the coherence time they possess full CSI, as shown in Fig. \ref{fig:1}. We define $\alpha\triangleq T_f/T_c$, as a measure for channel training efficiency, with $\alpha=0$ denoting full CSIT, and $\alpha=1$ representing delayed CSIT. This model is called \emph{mixed CSIT} in the literature \cite{Gou12}.

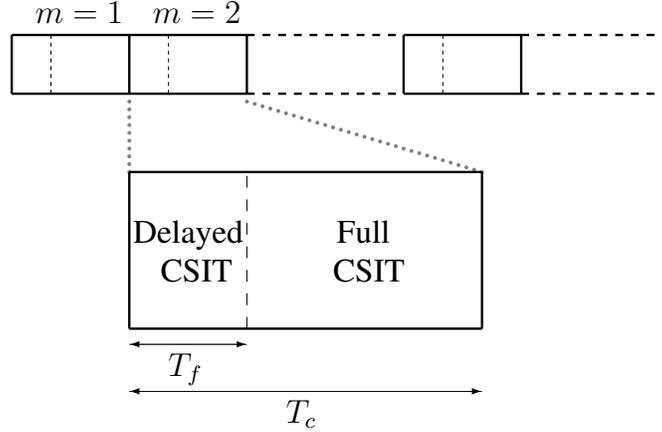
\begin{figure}
\centering
\resizebox{0.6\textwidth}{!}{
\setlength{\unitlength}{1cm}
\begin{picture}(10,5.5)(-0.7,-4.2)
\linethickness{0.3mm}

\multiput(0,0)(1.5,0){3}
{\line(0,1){0.75}}
\multiput(5,0)(1.5,0){2}
{\line(0,1){0.75}}

\multiput(0,0)(0,0.75){2}
{\line(1,0){3}}
\multiput(5,0)(0,0.75){2}
{\line(1,0){1.5}}

\multiput(3,0)(0.2,0){10}
{\line(1,0){0.1}}
\multiput(3,0.75)(0.2,0){10}
{\line(1,0){0.1}}

\multiput(6.5,0)(0.2,0){10}
{\line(1,0){0.1}}
\multiput(6.5,0.75)(0.2,0){10}
{\line(1,0){0.1}}

\linethickness{0.1mm}

\multiput(0.5,0)(0,0.1){8}
{\line(0,1){0.05}}
\multiput(2,0)(0,0.1){8}
{\line(0,1){0.05}}
\multiput(5.5,0)(0,0.1){8}
{\line(0,1){0.05}}

\put(0.3,0.9){$m=1$}
\put(1.8,0.9){$m=2$}

\multiput(1.5,-1)(0,0.1){10}{\color{gray}\circle*{0.05}}
\multiput(6,-1)(-0.1,0.03){31}{\color{gray}\circle*{0.05}}

\linethickness{0.3mm}

\multiput(1.5,-3)(4.5,0){2}{\line(0,1){2}}
\multiput(1.5,-3)(0,2){2}{\line(1,0){4.5}}

\linethickness{0.1mm}

\multiput(3,-3)(0,0.3){7}{\line(0,1){0.15}}
\put(1.5,-3.2){\vector(1,0){1.5}}
\put(3,-3.2){\vector(-1,0){1.5}}
\put(2,-3.6){$T_f$}
\put(1.5,-3.8){\vector(1,0){4.5}}
\put(6,-3.8){\vector(-1,0){4.5}}
\put(3.5,-4.2){$T_c$}

\put(1.55,-1.9){Delayed}
\put(1.9,-2.4){CSIT}

\put(4.15,-1.9){Full}
\put(4.1,-2.4){CSIT}

\end{picture}
}
\caption{Mixed CSIT availability in the block fading model.}
\label{fig:1}
\end{figure}

We consider a library of $N$ files, ${\cal W}=\{W_1,...,W_N\}$, of size $F$ bits each. We also assume that each transmitter has a cache of size $M_tF$ bits, and each receiver is equipped with a cache of size $M_rF$ bits. For later use, we define parameters $t_t \triangleq \frac{K_tM_t}{N}$ and $t_r \triangleq \frac{K_rM_r}{N}$. We, further, suppose that aggregate cache size of the transmitters is greater than the total size of all the $N$ files in the library, so the condition $t_t \ge 1$ should be satisfied.

The network operates in two phases, namely \emph{cache content placement} and \emph{content delivery}. In the former, the cache of transmitters and receivers are filled with data from the library, without prior knowledge of receivers' actual requests. In the latter, each receiver requests a single file from the library, where file $W_{d_k}$ is demanded by $k$th receiver. Then, accordingly, the transmitters cooperate to send proper data such that each receiver is able to retrieve its demanded file using the received data and the content stored in its cache.

In this paper, as the performance metric, we use  Degrees of Freedom (DoF) as a measure of the number of non-interfering parallel paths which can be used to deliver different data streams simultaneously (or equivalently in our case the number of users being served at the same time), defined as follows.

\begin{defi}
	For an interference channel consisting of $K_t$ transmitters, each with a cache of size $M_tF$ bits, and $K_r$ receivers, each with a cache of size $M_rF$ bits, the optimum DoF is defined as follows:
	\begin{equation}
	d^*(t_t,t_r) = \lim_{P\rightarrow \infty} \inf \frac{C_\Sigma(t_t,t_r,P)}{\log P},
	\end{equation}
	where $C_\Sigma(t_t,t_r,P)$ is the supremum of all achievable sum-rates received by all the receivers for normalized cache sizes of $t_t$ and $t_r$ on transmit and receive sides respectively, as well as a power constraint of $P$ held for transmitter antennas.
	
	According to this definition, the value of DoF cannot be greater than the number of all receivers in the system ($K_r$).
\end{defi}

\section{Cache-enabled interference channel with perfect instantaneous CSIT}
\label{sec:PerfectCSIT}

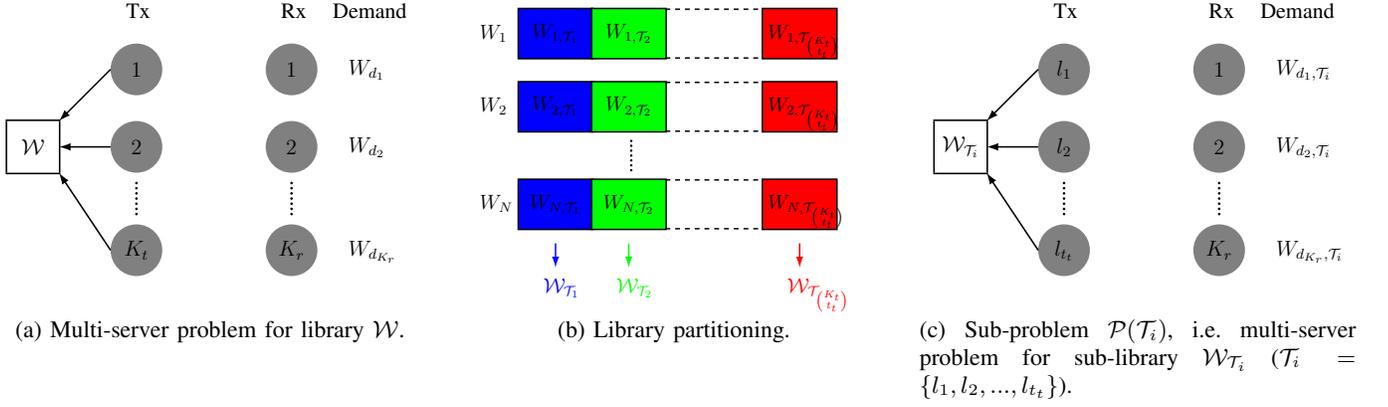
\begin{figure*}[t]
\begin{subfigure}[t]{.32\textwidth}
\resizebox{1\textwidth}{!}{
\setlength{\unitlength}{1cm}
\begin{picture}(8,6)(0,-4)
\setlength\fboxsep{0pt}
\linethickness{0.3mm}
\put(2.5,0.5){\color{gray}\circle*{1}}
\put(2,0){{\makebox(1,1){$1$}}}
\put(2.5,-1){\color{gray}\circle*{1}}
\put(2,-1.5){{\makebox(1,1){$2$}}}
\multiput(2.5,-2.3)(0,0.1){7}{\circle*{0.05}}
\put(2.5,-3){\color{gray}\circle*{1}}
\put(2,-3.5){{\makebox(1,1){$K_t$}}}

\put(5.5,0.5){\color{gray}\circle*{1}}
\put(5,0){{\makebox(1,1){$1$}}}
\put(5.5,-1){\color{gray}\circle*{1}}
\put(5,-1.5){{\makebox(1,1){$2$}}}
\multiput(5.5,-2.3)(0,0.1){7}{\circle*{0.05}}
\put(5.5,-3){\color{gray}\circle*{1}}
\put(5,-3.5){{\makebox(1,1){$K_r$}}}

\put(0,-1.5){\framebox(1,1){${\cal W}$}}

\put(2,0.5){\vector(-1,-1){1}}
\put(2,-1){\vector(-1,0){1}}
\put(2,-3){\vector(-2,3){1}}

\put(2.3,1.5){Tx}
\put(5.3,1.5){Rx}
\put(6.3,1.5){Demand}
\put(6.6,0.4){$W_{d_1}$}
\put(6.6,-1.1){$W_{d_2}$}
\put(6.6,-3.1){$W_{d_{K_r}}$}
\end{picture}
}
\caption{Multi-server problem for library $\mathcal{W}$.}
\label{fig:2a}
\end{subfigure}
~
\begin{subfigure}[t]{.32\textwidth}
\resizebox{1\textwidth}{!}{
\setlength{\unitlength}{1cm}
\begin{picture}(8.5,6)(-1,-5)
\setlength\fboxsep{0pt}
\linethickness{0.3mm}
\put(-0.8,0.4){$W_1$}
\put(0,0){\colorbox{blue}{\framebox(1.5,1){$W_{1,{{\cal T}_1}}$}}}
\put(1.5,0){\colorbox{green}{\framebox(1.5,1){$W_{1,{{\cal T}_2}}$}}}
\put(5,0){\colorbox{red}{\framebox(1.5,1){}}}
\put(5.1,0.4){$W_{1,{{\cal T}_{K_t \choose t_t}}}$}
\multiput(3,0)(0.2,0){10}{\line(1,0){0.1}}
\multiput(3,1)(0.2,0){10}{\line(1,0){0.1}}

\put(-0.8,-1.1){$W_2$}
\put(0,-1.5){\colorbox{blue}{\framebox(1.5,1){$W_{2,{{\cal T}_1}}$}}}
\put(1.5,-1.5){\colorbox{green}{\framebox(1.5,1){$W_{2,{{\cal T}_2}}$}}}
\put(5,-1.5){\colorbox{red}{\framebox(1.5,1){}}}
\put(5.1,-1.1){$W_{2,{{\cal T}_{K_t \choose t_t}}}$}
\multiput(3,-1.5)(0.2,0){10}{\line(1,0){0.1}}
\multiput(3,-0.5)(0.2,0){10}{\line(1,0){0.1}}

\multiput(2.3,-2.3)(0,0.1){7}{\circle*{0.05}}

\put(-0.8,-3.1){$W_N$}
\put(0,-3.5){\colorbox{blue}{\framebox(1.5,1){$W_{N,{{\cal T}_1}}$}}}
\put(1.5,-3.5){\colorbox{green}{\framebox(1.5,1){$W_{N,{{\cal T}_2}}$}}}
\put(5,-3.5){\colorbox{red}{\framebox(1.5,1){}}}
\put(5.1,-3.1){$W_{N,{{\cal T}_{K_t \choose t_t}}}$}
\multiput(3,-3.5)(0.2,0){10}{\line(1,0){0.1}}
\multiput(3,-2.5)(0.2,0){10}{\line(1,0){0.1}}

\put(0.75,-3.8){\color{blue}\vector(0,-1){0.5}}
\put(0.5,-4.8){\color{blue}${\cal W}_{{\cal T}_1}$}

\put(2.25,-3.8){\color{green}\vector(0,-1){0.5}}
\put(2,-4.8){\color{green}${\cal W}_{{\cal T}_2}$}

\put(5.75,-3.8){\color{red}\vector(0,-1){0.5}}
\put(5.5,-4.8){\color{red}${\cal W}_{{\cal T}_{{K_t \choose t_t}}}$}
\end{picture}
}
\caption{Library partitioning.}
\label{fig:2b}
\end{subfigure}
~
\begin{subfigure}[t]{.32\textwidth}
\resizebox{1\textwidth}{!}{
\setlength{\unitlength}{1cm}
\begin{picture}(8,6)(0,-4)
\setlength\fboxsep{0pt}
\linethickness{0.3mm}
\put(2.5,0.5){\color{gray}\circle*{1}}
\put(2,0){{\makebox(1,1){$l_1$}}}
\put(2.5,-1){\color{gray}\circle*{1}}
\put(2,-1.5){{\makebox(1,1){$l_2$}}}
\multiput(2.5,-2.3)(0,0.1){7}{\circle*{0.05}}
\put(2.5,-3){\color{gray}\circle*{1}}
\put(2,-3.5){{\makebox(1,1){$l_{t_t}$}}}

\put(5.5,0.5){\color{gray}\circle*{1}}
\put(5,0){{\makebox(1,1){$1$}}}
\put(5.5,-1){\color{gray}\circle*{1}}
\put(5,-1.5){{\makebox(1,1){$2$}}}
\multiput(5.5,-2.3)(0,0.1){7}{\circle*{0.05}}
\put(5.5,-3){\color{gray}\circle*{1}}
\put(5,-3.5){{\makebox(1,1){$K_r$}}}

\put(0,-1.5){\framebox(1,1){${\cal W}_{{\cal T}_i}$}}

\put(2,0.5){\vector(-1,-1){1}}
\put(2,-1){\vector(-1,0){1}}
\put(2,-3){\vector(-2,3){1}}

\put(2.3,1.5){Tx}
\put(5.3,1.5){Rx}
\put(6.3,1.5){Demand}
\put(6.6,0.4){$W_{d_1,{\cal T}_i}$}
\put(6.6,-1.1){$W_{d_2,{\cal T}_i}$}
\put(6.6,-3.1){$W_{d_{K_r},{\cal T}_i}$}
\end{picture}
}
\caption{Sub-problem ${\cal P}({\cal T}_i)$, i.e. multi-server problem for sub-library $\mathcal{W}_{\mathcal{T}_i}$ (${\cal T}_i=\{l_1,l_2,...,l_{t_t}\}$).}
\label{fig:2c}
\end{subfigure}
\caption{The general $M_t$ case can be reduced to ${K_t \choose t_t}$ multi-server sub-problems with smaller sub-libraries.}
\label{fig:2}
\end{figure*}

In this section, we derive DoF when full CSIT is available. First, we consider the special case of $M_t=N$, which is stated in Lemma \ref{lem:1}. Then, the main result of this section for general $M_t$ is presented in Theorem \ref{thm:0}. 

\begin{lem}[Multi-server problem ($M_t=N$), \cite{Shariatpanahi16}]
	\label{lem:1}
	Suppose an interference channel with a transfer matrix of i.i.d. elements, and a library of $N$ files each of size $F$ bits. If there are $K_t$ transmitters each with a cache of size $NF$ bits, and $K_r$ receivers each with a cache of size $M_rF$ bits, the following DoFs (and their lower convex envelope) are achievable:
	\begin{equation}
	d(K_t,t_r) = \min(K_t+t_r,K_r),
	\end{equation}
	where $t_r=\frac{K_rM_r}{N} \in \{0,1,...,K_r\}$ should be held true.
\end{lem}
\begin{proof}
	The proof is given in \cite{Shariatpanahi16}.
\end{proof}

The main point of Lemma \ref{lem:1} is that each transmitter has access to the full library as depicted in Fig. \ref{fig:2a}. Next, in Theorem \ref{thm:0}, we generalize Lemma \ref{lem:1} to the case where transmitters have limited memories, such that no transmitter is able to cache the entire library.

\begin{thm}
	\label{thm:0}
	Suppose an interference channel with a transfer matrix of i.i.d. elements, and a library of $N$ files each of size $F$ bits. If there are $K_t$ transmitters each with a cache of size $M_tF$ bits,  and $K_r$ receivers each with a cache of size $M_rF$ bits, the following DoFs (and their lower convex envelope) are achievable:
	\begin{equation}
	d_{\text{C}}(t_t,t_r) = \min(t_t+t_r,K_r),
	\end{equation}
	where $t_t=\frac{K_tM_t}{N} \in \{1,2,...,K_t\}$ and $t_r=\frac{K_rM_r}{N} \in \{0,1,...,K_r\}$ should be held true.
\end{thm}
\begin{proof}
	First, we divide each file into ${K_t \choose t_t }{K_r \choose t_r}$ equal-sized non-overlapping sub-files:
	\begin{equation}
	W_n = (W_{n,{\cal T},{\cal R}}:{\cal T} \subseteq [K_t],|{\cal T}| = t_t,{\cal R} \subseteq [K_r] , |{\cal R}| = t_r ).
	\end{equation}
	Then, throughout the cache content placement phase, receiver $k$ caches   sub-file $W_{n,{\cal T},{\cal R}}$ if $k \in {\cal R}$, and transmitter $l$ caches  sub-file $W_{n,{\cal T},{\cal R}}$ if $l \in {\cal T}$, as follows:
	\begin{eqnarray} \nonumber
	Z_{t_l} &=& (W_{n,{\cal T},{\cal R}}:l \in {\cal T}), \\ 
	Z_{r_k} &=& (W_{n,{\cal T},{\cal R}}:k \in {\cal R}),
	\end{eqnarray}
	where $Z_{t_l}$ and $Z_{t_r}$ are referred to as the content of transmitter $l$'s cache and the receiver $k$'s cache, respectively.
	
	It can be seen that 
	$W_{n,{\cal T}}=(W_{n,{\cal T},{\cal R}}:{\cal R} \subseteq [K_r] , |{\cal R}| = t_r)$
	is present in the caches of exactly $t_t$ transmitters, i.e. the ones which belong to ${\cal T}$. So, we can divide the primary library $\mathcal{W}$ into ${K_t \choose t_t }$ sub-libraries as follows:
	\begin{equation}
	{\cal W}_{\cal T} = \{W_{1,{\cal T}},...,W_{N,{\cal T}}\}, \quad {\cal T}\subseteq [K_t],|{\cal T}|=t_t,
	\end{equation}
	each of which is completely cached in precisely $t_t$ transmitters, and also is comprised of $N$ smaller files each of size $F'=\frac{F}{{K_t \choose t_t}}$ bits (see Fig. \ref{fig:2b}).
	
	Let us consider one of these sub-libraries, say $\mathcal{W}_{\mathcal{T}_i}$. We define sub-problem ${\cal P}(\mathcal{T}_i)$ to be the problem of delivering $(W_{d_k,\mathcal{T}_i})_{k=1}^{K_r}$ by $t_t$ transmitters (members of $\mathcal{T}_i$) all having access to the whole sub-library $\mathcal{W}_{\mathcal{T}_i}$, as depicted in Fig. \ref{fig:2c}. Consequently, if the necessary transmissions of all sub-problems $\big({\cal P}(\mathcal{T}_i)\big)_{i=1}^{{K_t \choose t_t}}$ are fulfilled in sequence, each receiver $k$ will be able to decode $W_{d_k}$. This reduces the original problem to ${K_t \choose t_t}$ multi-server problems with smaller sub-libraries.
	
	 Then, the result of Lemma \ref{lem:1} can be applied with $t_t$ transmitters to each sub-problem, which states that, for each $\mathcal{T}_i$, the delivery of $(W_{d_k,\mathcal{T}_i})_{k=1}^{K_r}$ can be accomplished with a DoF of
	 
	\begin{equation}
	d(t_t,t_r) = \min(t_t+t_r,K_r).
	\end{equation}  
	 Since sub-problems are handled in sequence, the same DoF is achievable for the whole library, and this concludes the proof.
\end{proof}

It should be noted that a similar result, for the full CSIT case, has also been obtained independently in \cite{Naderalizadeh16}, via a different proof approach without relying on Lemma \ref{lem:1}. Since, in the next section, we need more details of the delivery scheme, we have provided the delivery details in the \nameref{app}.

\section{Cache-enabled interference channel with mixed CSIT}
\label{sec:DelayedCSIT}

The interesting result of \cite{dscit-maddah} shows that  delayed CSIT can still be of much benefit in order to achieve greater DoF than the one without CSIT. Considering a MISO broadcast channel in which transmitters access channel coefficients of all preceding time slots excluding the present time slot, the authors in \cite{dscit-maddah} propose a novel method which employs such past knowledge to achieve an optimum DoF of the channel. We will state the main theorem of the aforementioned work in the following lemma.
\begin{lem}[\cite{dscit-maddah}]
	\label{lem:2}
	For a MISO broadcast channel with $L$ transmit antennas, $K$ receivers, and one-time-slot-delayed CSIT, for $L\ge K-j+1$ we have:
	\begin{equation}
	d_j^*(L,K) \ge \frac{K-j+1}{\frac{1}{j}+\frac{1}{j+1}+...+\frac{1}{K}},
	\end{equation}
	where $d_j^*(L,K)$ denotes the maximum DoF to transfer messages of order $j$ (messages which are intended for exactly $j$ receivers).
\end{lem}

Based on Lemma \ref{lem:2}, we can state the following result for the delayed CSIT cache-enabled interference channel.

\begin{corol}\label{corol:1}
\emph{In a cache-enabled $K_t \times K_r$ interference channel with delayed CSIT the following DoFs (and their lower convex envelope) are achievable:
	\begin{equation}
	d_{\text{D}}(t_t,t_r) = \frac{t_t}{\frac{1}{t_r+1}+\frac{1}{t_r+2}+...+\frac{1}{t_r+t_t}}.
	\end{equation}
	where $t_t=\frac{K_tM_t}{N} \in \{1,2,...,K_t\}$ and $t_r=\frac{K_rM_r}{N} \in \{0,1,...,K_r\}$, and $t_t+t_r \le K_r$ should be held true.}
\end{corol}

\begin{proof}
	Choosing a $t_t$-subset of transmitters, ${\cal T}$, and a subset of receivers of cardinality $t_t+t_r$, ${\cal S}$, according to the procedure introduced in \nameref{app}, results in a pair $({\cal T},{\cal S})$ such that messages of order $t_r+1$ (each beneficial to a $(t_r+1)$-subset of ${\cal S}$) can be generated using the content which is present in the caches of transmitters in ${\cal T}$. Drawing the analogy between the aforementioned set-up and the one in Lemma \ref{lem:2}, we observe that DoF order $t_r+1$ can be calculated from above lemma substituting $t_r+1$, $t_t$, and $t_t+t_r$ for $j$, $L$, and $K$ respectively:
	\begin{equation}
	d_{(t_r+1)}^*(t_t,t_t+t_r) \ge  \frac{t_t}{\frac{1}{t_r+1}+\frac{1}{t_r+2}+...+\frac{1}{t_r+t_t}}.
	\end{equation}
	Thus, the procedure described in Algorithm \ref{alg:1} results in the desired DoF.
\end{proof}

\RestyleAlgo{boxruled}
\LinesNumbered
\IncMargin{1em}
\begin{algorithm}[ht]
	\caption{\sc{Delayed CSIT Delivery} \label{alg:1}}
	$t_t \leftarrow K_tM_t/N$\\
	$t_r \leftarrow K_rM_r/N$\\
	$i \leftarrow 1$\\
	\ForEach{${\cal T} \subseteq [K_t], |{\cal T}|=t_t$}{
		\ForEach{${\cal S} \subseteq [K_r], |{\cal S}|=t_r+t_t$}{
			\ForEach{${\cal R} \subseteq {\cal S}, |{\cal R}|=t_r+1$}{
				$U_i \leftarrow \oplus_{r\in \mathcal{R}} W_{d_r,{\cal T},\mathcal{R}\setminus\{r\}}$\\
				$i \leftarrow i+1$
			}
			${\bf U} \leftarrow \Big(U_1,...,U_{t_t+t_r \choose t_r+1}\Big)$\\
			MAT$({\cal T},{\cal S},t_r+1,{\bf U}) ^*$
		}
	}
	\nonl \hrulefill\\
	\nonl \small{*MISO BC delivery procedure in \cite{dscit-maddah} with transmitters ${\cal T}$ and receivers ${\cal S}$ to transfer messages in {\bf U} of order $t_r+1$}
\end{algorithm}
\DecMargin{1em}

Now, we proceed to state the main theorem that introduces the achievable DoF which can be attained in case that transmitters experience a feedback channel resulting in mixed CSIT.

\begin{thm}
	\label{thm:1}
	For an interference channel comprising $K_t$ transmitters and $K_r$ receivers, each transmitter with a cache of size $M_tF$ bits and each receiver equipped with a $M_rF$-bit memory, and a library including $N$ files each of size $F$ bits, if the state of CSIT availability is characterized by a factor $0 \le \alpha \le 1$, the following DoF is achievable:
		\begin{equation} \label{eq:5}
		d_{\text{M}}(t_t,t_r) = t_t+t_r+\alpha\Big( \frac{ t_t}{\frac{1}{t_r+1}+\frac{1}{t_r+2}+...+\frac{1}{t_r+t_t}}-t_t-t_r\Big),\\
		\end{equation} 
	where $t_t=\frac{K_tM_t}{N} \in \{1,2,...,K_t\}$, $t_r=\frac{K_rM_r}{N} \in \{0,1,...,K_r\}$, and $t_t+t_r \le K_r$ should be held true.  
\end{thm}

\begin{proof}
	Availability of mixed CSIT characterized by the parameter $\alpha$ shows that during $\alpha$ portion of time, transmitters access CSIT with a delay of one time slot, and have  current CSIT for the rest of time.
	
	Thus, in each time coherence block $T_c$, during the training phase Algorithm \ref{alg:1} is used for delivery, and, in the remaining period Algorithm \ref{alg:2} is used. Therefore, it is clear that the total DoF equals the linear combination of the ones in Theorem \ref{thm:0} and Corollary \ref{corol:1} with coefficients $(1-\alpha)$ and $\alpha$ respectively, which leads to \eqref{eq:5}.
\end{proof}

\section{Numerical Illustrations and Conclusions}
\label{sec:Num}

In this section, we examine the result of Theorem \ref{thm:1} in order to gain an insight into how each of parameters which describe the system model contributes to the achievable DoF in \eqref{eq:5}. To this end, we consider two main parameters of the system and investigate the interaction of achievable DoF and these parameters. The first one is the portion of each time block during which perfect CSIT is not available denoted by $\alpha$ as we previously defined. The latter one is the portion of all cache size allocated to the receivers to the total cache size in the transmit side, which we denote by $\beta = \frac{t_r}{t_t}$.

Fig. \ref{fig:3} represents the impact of $\alpha$ on the achievable DoF in \eqref{eq:5}. As shown in the diagram, we can see that the DoF is inversely related to $\alpha$ in a linear manner, for fixed cache sizes at both transmit and receive sides. For a fixed $\beta$, it also can be figured out that increasing total cache size of system not only leads to a greater DoF, but also steepens the slope of the diagram of DoF versus $\alpha$, which means that the more the system aggregate cache size is, the more the influence of full CSIT on enhancing DoF will be.  

\begin{figure}[h]
	\centering
	\includegraphics[width=0.7\textwidth]{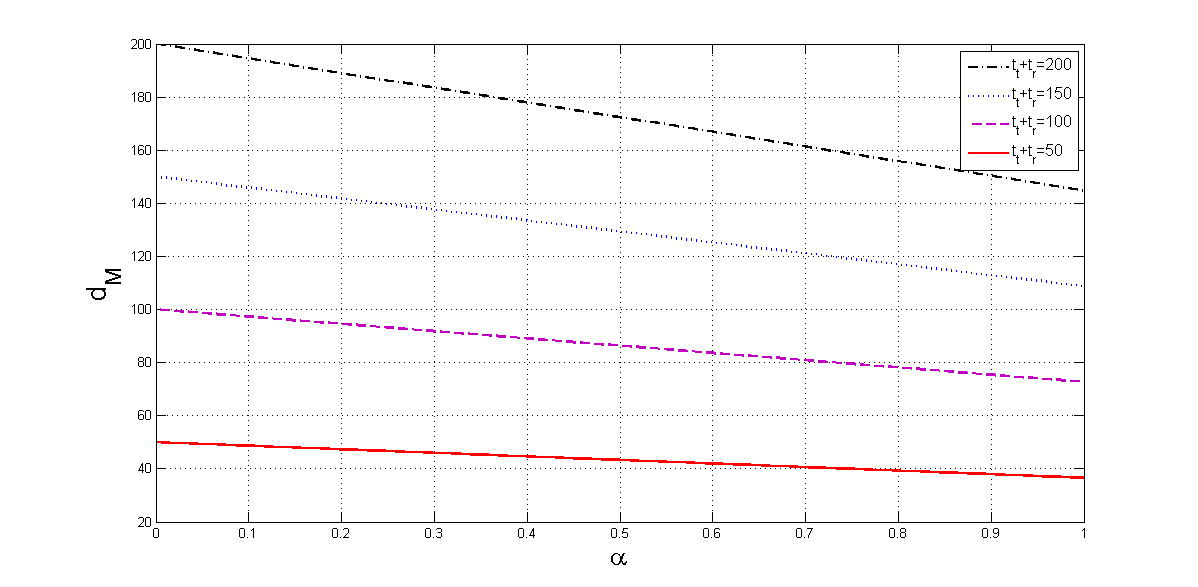}
	\caption{DoF for mixed CSIT, $d_M$, versus CSIT availability amount, $\alpha$, for fixed $\beta=1$ and different aggregate cache sizes  $t_t+t_r$.}
	\label{fig:3}
	\includegraphics[width=0.7\textwidth]{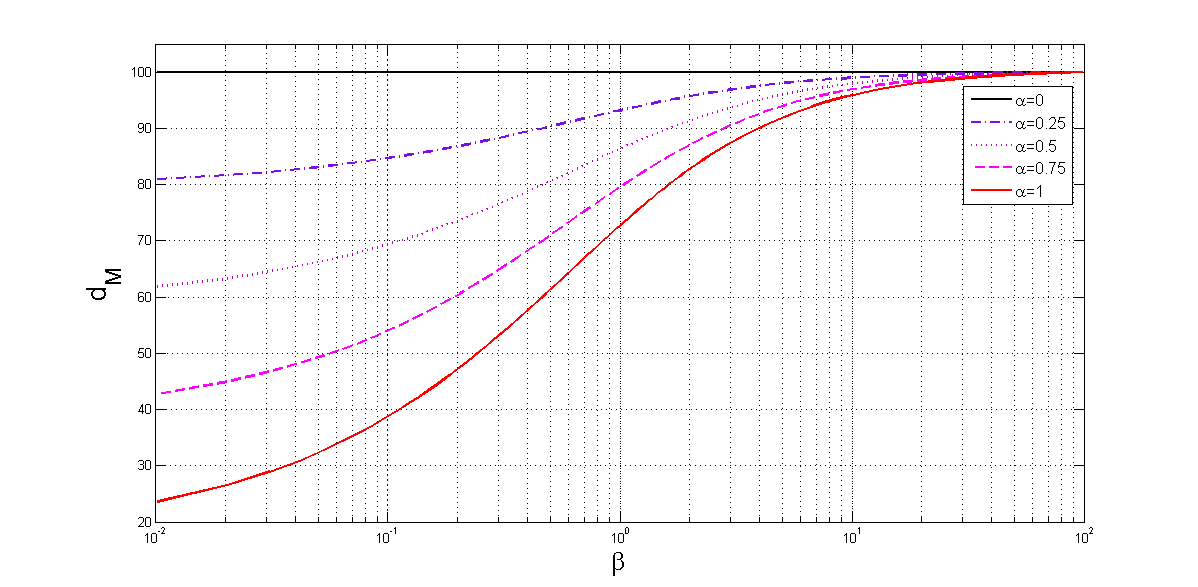}
	\caption{DoF for mixed CSIT, $d_M$, versus the ratio of cache size on receive side to the one on transmit side, $\beta$, for fixed  $t_t+t_r=100$, and different CSIT availability values $\alpha$.}
	\label{fig:4}
\end{figure}

Also, the effect of $\beta$ on DoF is shown in Fig. \ref{fig:4}. This figure indicates how the caches at receive side can be more beneficial in practical conditions than the caches at transmit side. In this figure, we assume a fixed normalized total cache size $t_t+t_r=100$. The case $\alpha=0$ refers to the setup in \cite{Naderalizadeh16} which indicates equal value for allocating memory at either transmit or receive side for the current CSIT case. For $\alpha>0$, i.e. mixed CSIT case, by increasing $\beta$ (allocating more storage to the receive side), DoF will increase.

Therefore, considering an interference channel with caches at both transmit and receive sides, the present study demonstrated that equipping receivers with caches benefits content delivery phase in terms of DoF more than doing so for transmitters when a more practical scenario is encountered. In particular, we assumed a pragmatic approach of providing CSIT including a training/feedback process, leading to a combination of current and delayed CSIT availability, namely mixed CSIT. We introduced an achievable DoF in such setting that accentuates the superiority of receive-side caches over transmitters' memories for lower qualities of feedback.


\newpage
\appendix{Proof of Theorem \ref{thm:0} with Detailed Delivery Procedure}
\label{app}
\begin{proof}
	Suppose receiver $k$, which has demanded $W_{d_k}$ in the second phase, and all the subfiles of form $W_{d_k,{\cal T, }{\cal R}}$, where $k \not \in {\cal R}$ should be delivered to this receiver. Note that $W_{d_k,{\cal T, }{\cal R}}$ is available in:
	\begin{itemize}
	\item[1.] all the $t_r$ receivers $k \in {\cal R}$, and
	\item[2.] all the $t_t$ transmitters $l \in {\cal T}$.
	\end{itemize}
	If $W_{d_k,{\cal T, }{\cal R}}$ is broadcasted to all the receivers cooperatively by $t_t$ transmitters $l \in {\cal T}$, then 
	\begin{itemize}
	\item[1.] receiver $k$ will benefit from this subfile to decode $W_{d_k}$,
	\item[2.] $t_r$ receivers in $ {\cal R}$ can eliminate this interference with the help of their cache contents,
	\item[3.] this interference can be zero-forced at an arbitrary subset of receivers of size $t_t-1$, and
	\item[4.] this will inevitably cause interference for the other $K_r-(t_r+t_t)$ receivers.
	\end{itemize}	 	
	From the above observations, one can see that it is possible to design coded messages which are useful for $t_t+t_r$ receivers leading to DoF of $t_t+t_r$.
	
	In order to design such common message, we further divide each sub-file into ${K_r-t_r-1 \choose t_t-1}$ mini-files, such that during the transmission of each mini-file, interference cancellation will be performed for a definite subset of receivers that neither have the mini-file in their caches, nor need to receive the mini-file to restore their demands:
	\begin{equation}
	W_{n,{\cal T},{\cal R}} = (W_{n,{\cal T},{\cal R}}^{\cal R'}:{\cal R'} \subseteq [K_r-t_r-1] , |{\cal R'}| = t_t-1 ).
	\end{equation}
	
	Next, we define an arbitrary bijective function $f_{K_r}^{\cal U}:[K_r-|{\cal U}|]\rightarrow [K_r]\setminus {\cal U}$ for each ${\cal U} \subseteq [K_r]$, thus to send mini-file $W_{d_k,{\cal T},{\cal R}}^{\cal R'}$ to the receiver $k$, we perform a pre-coding through the transmitters in ${\cal T}$ to zero-force the interference at the receivers in $f_{K_r}^{{\cal R}\cup \{k\}}({\cal R'})$, which denotes  a $(t_t-1)$-subset of receivers and is defined as follows:
	\begin{equation}
	f_{K_r}^{{\cal R}\cup \{k\}}({\cal R'}) = \{f_{K_r}^{{\cal R}\cup \{k\}}(k') :k' \in {\cal R'} \}.
	\end{equation}
	
	In the content delivery phase, we claim that it is possible to benefit a set of $t_t+t_r$ receivers simultaneously. To explain this, let's consider a subset ${\cal S} \subseteq [K_r]$ of cardinality $t_t+t_r$ and a $t_t$-subset ${\cal T} \subseteq [K_t]$. For this specific ${\cal S}$, a $(t_t+t_r)$-subset of receivers, we denote all possible subsets of ${\cal S}$ of cardinality $t_r+1$ by ${\cal S}_i,i=1,...,{t_t+t_r \choose t_r+1}$, which means ${\cal S}_i \subseteq {\cal S}$ and $|{\cal S}_i|=t_r+1$. We then determine a $t_t$-by-$1$ vector ${\bf u}_{{\cal T},{\cal S}}^{{\cal S}_i}$ for each ${\cal S}_i$ such that makes the below statements hold true:
	\begin{eqnarray}
	\label{eq:3}
	{\bf u}_{{\cal T},{\cal S}}^{{\cal S}_i} &\perp& {\bf h}_j^{\cal T}, \quad \forall j \in {\cal S}\setminus {\cal S}_i, \nonumber\\
	{\bf u}_{{\cal T},{\cal S}}^{{\cal S}_i} &\not\perp& {\bf h}_j^{\cal T}, \quad \forall j \in {\cal S}_i,
	\end{eqnarray}
	where ${\bf h}_j^{\cal T} \triangleq [h_{j,l_1},...,h_{j,l_{|{\cal T}|}}]^T$ is defined for any ${\cal T}=\{l_1,...,l_{|{\cal T}|}\} \subseteq [K_t]$.
	
	It can be shown that the above conditions are true with high probability if the size of the field from which transmitted symbols are chosen is much greater than $(t_r+1){K_r \choose t}{t \choose t_r+1}$ \cite{Shariatpanahi16}.
	
	Next, for each ${\cal S}_i$ define:
	\begin{equation}
	G_{\cal T}({\cal S}_i) = L_{r \in {\cal S}_i} \Big(W_{d_r,{\cal T},{\cal S}_i\setminus \{r\}}^{{\cal R}_i}\Big),
	\end{equation}
	with ${\cal R}_i$ is required to be such that $f_K^{{\cal S}_i}({\cal R}_i) = {\cal S}\setminus {\cal S}_i$, and $ L_{r \in {\cal S}_i}$ represents a random linear combination of the corresponding mini-files for all $r \in {\cal S}_i$. In addition, $W_{d_r,{\cal T},{\cal S}_i\setminus \{r\}}^{{\cal R}_i}$ is a mini-file of the file demanded by receiver $r$, which exists in the caches of transmitters in ${\cal T}$ and receivers in ${\cal S}_i\setminus \{r\}$, and we need to perform zero-forcing to cancel its interference to the receivers in ${\cal S}\setminus {\cal S}_i$.
	
	Next, for such $t_t$-subset ${\cal T}$ and $(t_t+t_r)$-subset ${\cal S}$ we define:
	\begin{equation}\label{eq:commonmessage}
	{\bf X}_{\cal T}({\cal S}) = \sum_{{\cal S}_i \subseteq {\cal S},|{\cal S}_i|=t_r+1} {\bf u}_{{\cal T},{\cal S}}^{{\cal S}_i}G_{\cal T}({\cal S}_i).
	\end{equation}
	
	We repeat the aforementioned procedure ${t_t+t_r-1 \choose t_r}$ times with different random coefficients to obtain different $t_t$-by-$1$ vectors ${\bf X}_{\cal T}^{\omega}({\cal S}), \omega=1,...,{t_t+t_r-1 \choose t_r}$ as follows:
	\begin{eqnarray}
	\label{eq:4}
	G_{\cal T}^\omega({\cal S}_i) &=& L_{r \in {\cal S}_i}^\omega \Big(W_{d_r,{\cal T},{\cal S}_i\setminus \{r\}}^{{\cal R}_i}\Big), \\
	{\bf X}_{\cal T}^\omega({\cal S}) &=& \sum_{{\cal S}_i \subseteq {\cal S},|{\cal S}_i|=t_r+1} {\bf u}_{{\cal T},{\cal S}}^{{\cal S}_i}G_{\cal T}^\omega({\cal S}_i),
	\end{eqnarray}
	where the random coefficients for calculating linear combinations in $L_{r \in {\cal S}_i}^\omega$'s are different for different $\omega$'s in order to achieve independent linear combinations of the corresponding mini-files, with high probability.
	
	Then, for these ${\cal T}$ and ${\cal S}$ sets, the transmitters in ${\cal T}$ transmit the following block:
	\begin{equation}
	\Big[{\bf X}_{\cal T}^1({\cal S}) ,..., {\bf X}_{\cal T}^{t_t+t_r-1 \choose t_r}({\cal S}) \Big].
	\end{equation}
	After accomplishing the transmission of the above block, a different pair $({\cal T},{\cal S})$ should be chosen (such that ${\cal T}\subseteq[K_r],|{\cal T}|=t_t$, and ${\cal S}\subseteq[K_t],|{\cal S}|=t_t+t_r$), and the same procedure is performed until all such pairs have been taken into account.
	
	It can be shown that, after the above procedure is concluded, all the receivers will be able to decode their requested files. Also, since each message in \eqref{eq:commonmessage} benefits $t_t+t_r$ receivers simultaneously, achievable DoF of this delivery scheme is $\min(t_t+t_r,K_r)$. The procedure to transmit the desired files is shown in Algorithm \ref{alg:2}.
\end{proof}

	\RestyleAlgo{boxruled}
	\LinesNumbered
	\IncMargin{1em}
	\begin{algorithm}[ht]
		\caption{{\sc Current CSIT Delivery} \label{alg:2}}
		$t_t \leftarrow K_tM_t/N$\\
		$t_r \leftarrow K_rM_r/N$\\
		\ForEach{${\cal T} \subseteq [K_t], |{\cal T}|=t_t$}{
			\ForEach{${\cal S} \subseteq [K_r], |{\cal S}|=t_r+t_t$}{
				\ForEach{${\cal R} \subseteq {\cal S}, |{\cal R}|=t_r+1$}{
					Design ${\bf u}_{{\cal T},{\cal S}}^{{\cal R}}$ such that: for all $j\in {\cal S}$, ${\bf h}_j^{\cal T} \perp {\bf u}_{{\cal T},{\cal S}}^{{\cal R}}$ if $j \notin {\cal R}$, and ${\bf h}_j^{\cal T} \not\perp {\bf u}_{{\cal T},{\cal S}}^{{\cal R}}$ if $j \in {\cal R}$
				}
				\ForEach{$\omega = 1,...,{t_t+t_r-1 \choose t_r}$}{
					\ForEach{${\cal R} \subseteq {\cal S}, |{\cal R}|=t_r+1$}{
					${\cal R}' \leftarrow {f_{K_r}^{{\cal R}}}^{-1}({\cal S}\setminus {\cal R})$\\
						$G_{\cal T}^\omega({\cal R}) = L_{r \in {\cal R}}^\omega \Big(W_{d_r,{\cal T},{\cal R}\setminus \{r\}}^{{\cal R}'}\Big)$
					}
					${\bf X}_{\cal T}^\omega({\cal S}) \leftarrow \sum_{{\cal R} \subseteq {\cal S},|{\cal R}|=t_r+1} {\bf u}_{{\cal T},{\cal S}}^{{\cal R}}G_{\cal T}^\omega({\cal R})$
				}
				transmit ${\bf X}_{\cal T}({\cal S}) = \Big[{\bf X}_{\cal T}^1({\cal S}) ,..., {\bf X}_{\cal T}^{t_t+t_r-1 \choose t_r}({\cal S}) \Big]$
			}
		}
	\end{algorithm}
	\DecMargin{1em}

\begin{thebibliography}{99}	
\bibitem{Cisco}
Cisco, ``Cisco visual networking index: Forecast and methodology, 2015-2020,'' \emph{White Paper}, Feb. 2016.

\bibitem{Pathan07}
A. M. K. Pathan and R. Buyya, ``A taxonomy and survey of content delivery networks,'' \emph{Grid Computing and Distributed Systems Laboratory, University of Melbourne, Technical Report}, 2007.

\bibitem{Shanmugam13}
K. Shanmugam, N. Golrezaei,  A. G. Dimakis, A F. Molisch, and G. Caire, ``FemtoCaching: Wireless content delivery through distributed caching helpers,'' \emph{IEEE Transactions on Information Theory}, vol. 59, no. 12, 2013, pp. 8402--8413.

\bibitem{Golrezaei14}
N. Golrezaei, A G. Dimakis, and A. F. Molisch, ``Scaling behavior for Device-to-Device communications with distributed caching,'' \emph{IEEE Transactions on Information Theory}, vol. 60, no. 7, 2014, pp. 4286--4298.

\bibitem{Bastug14}
E. Bastug, M. Bennis, and M. Debbah, ``Living on the edge: The role of proactive caching in 5G wireless networks,'' \emph{IEEE Communications Magazine}, vol. 52, no. 8, 2014, pp. 82--89.

\bibitem{Perabathini15}
B. Perabathini, E. Baştuğ, M. Kountouris, M. Debbah, and A. Conte, ``Caching at the edge: A green perspective for 5G networks,'' \emph{in Proc.  IEEE International Conference on Communication Workshop}, 2015.

\bibitem{Gitzenis13}
S. Gitzenis, G. S Paschos, and L. Tassiulas, ``Asymptotic laws for joint content replication and delivery in wireless networks,'' \emph{IEEE Transactions on Information Theory}, vol. 59, no. 5, 2013, pp. 2760--2776.

\bibitem{Maddah-Ali14}
M.A. Maddah-Ali and U. Niesen, ``Fundamental limits of caching,'' \emph{IEEE Transactions on Information Theory}, vol. 60, no. 5, 2014, pp. 2856--2867.

\bibitem{Ji16} 
M. Ji, G. Caire, and A. F. Molisch, ``Fundamental limits of caching in wireless D2D networks,'' \emph{IEEE Transactions on Information Theory}, vol. 62, no. 2, 2016, pp. 849--869.

\bibitem{Karamchandani16}
N. Karamchandani, U. Niesen, M. A. Maddah-Ali, and S. N. Diggavi, ``Hierarchical coded caching,'' \emph{IEEE Transactions on Information Theory}, vol. 62, no. 6, 2016, pp. 3212--3229.


\bibitem{Shariatpanahi16}
S. P. Shariatpanahi, S. A. Motahari, and B. Hossein Khalaj, ``Mutli-server coded caching,'' \emph{IEEE Transactions on Information Theory}, vol. 62, no. 12, 2016, pp. 7253--7271.

\bibitem{Naderalizadeh16}
N. Naderializadeh, M. A. Maddah-Ali, and A. S. Avestimehr, ``Fundamental limits of cache-aided interference management,'' arXiv preprint, arXiv:1602.04207, Feb.
2016.

\bibitem{Zhang15}
J. Zhang and P. Elia, ``Fundamental limits of cache-aided wireless BC: Interplay of coded-caching and CSIT feedback,'' arXiv preprint, arXiv:1511.03961, 2015.

\bibitem{Zhang16}
J. Zhang and P. Elia, ``The synergistic gains of coded caching and delayed feedback,'' arXiv preprint, arXiv:1604.06531, 2016.

\bibitem{Love08}
D. Love, R. Heath, V. Lau, D. Gesbert, B. Rao, and M. Andrews,
``An overview of limited feedback in wireless communication systems,''
\emph{IEEE Journal on Selected Areas in Communications}, vol. 26, no. 8, pp. 1341--1365, Oct. 2008.

\bibitem{Shariatpanahi14}
S. P. Shariatpanahi, H. Shah-Mansouri, and B. Hossein Khalaj, ``Caching gain
in wireless networks with fading: A multi-user diversity perspective,'' \emph{in Proc. IEEE WCNC}, 2014.

\bibitem{Maddah-Ali15}
M. A. Maddah-Ali and U. Niesen, ``Cache-aided interference channels,''
\emph{in Proc. of the IEEE International Symposium on Information
Theory (ISIT 2015)}, Hong-Kong, China, 2015.

\bibitem{Wigger16}
M. A. Wigger, R. Timo, and S. Shamai, ``Complete interference mitigation through receiver-caching in wyner’s networks,'' arXiv preprint, arXiv:1605.03761, 2016.

\bibitem{Gou12}
T. Gou and S. A. Jafar, ``Optimal use of current and outdated channel state information: Degrees of freedom of the MISO BC with mixed CSIT,'' \emph{IEEE Communications Letters}, vol. 16, no. 7, 2012, pp. 1084--1087.

\bibitem{dscit-maddah}M. A. Maddah-Ali and D. Tse, ``Completely stale transmitter channel state information is still very useful,'' \emph{IEEE Transactions on Information Theory}, vol. 58, no. 7, 2012, pp. 4418--4431.
\end{thebibliography}
\end{document}